\documentclass[final]{ws-ijac}

\usepackage{amsmath,amsthm,amssymb}
\usepackage{tikz, tikz-cd}
\usepackage{subcaption}
\usepackage{float}
\newcommand{\A}{\mathcal{A}}
\newcommand{\G}{\mathcal{G}}
\renewcommand{\P}{\mathfrak{A}}
\newcommand{\C}{\mathfrak{C}(\Am,\e)}
\newcommand{\Z}{\mathbb{Z}}
\newcommand{\Q}{\mathbb{Q}}
\newcommand{\2}{\textbf{2}}
\newcommand{\Am}{\textbf{A}}
\newcommand{\del}{\partial}
\newcommand{\vv}{\bar{v}}
\newcommand{\e}{\bar{e}}

\newtheorem{thm}{Theorem}

\title{Fractional Elements in Abelian Automata Groups}
\author{Chris Grossack}

\address%
{%
  Carnegie Mellon University, Pittsburgh, USA\\
  \email{cgrossac@alumni.cmu.edu}
}

\begin{document}
\maketitle

\begin{abstract}
  A theorem of Nekrashevych and Sidki shows the Mealy Automata
  structures one can place on $\Z^m$ are parametrized by a family of
  matrices (called ``$\frac{1}{2}$-integral'') and a choice of 
  residuation vector $\e \in \Z^m$. While the impact of the chosen 
  matrix is well understood, the impact of the residuation vector on the
  resulting structure was seemingly sporadic. 

  In this paper, we characterize the impact of the residuation vector $\e$
  by recognizing an initial structure when $\e$ is the first standard basis 
  vector. All other choices of $\e$ extend this initial structure by adding
  ``fractional elements'' in a way we make precise. 

  \keywords{Abelian Automata \and Group Theory \and Transducer \and Module Theory}
\end{abstract}

\section{Background}
Finite State Automata are combinatorial objects which encode relations 
between words over some alphabet. Automata provide deep connections between
combinatorics, algebra, and logic, and are essential tools in contemporary 
computer science. One such link is in the decidability of truth in a structure
whose relations are all computable by automata. One can combine these automata 
into more complicated automata representing logical sentences in such a way 
that a sentence is true if and only if a simple reachability condition holds
\cite{Brny07:automatic_structures}. This gives a simple proof that the theory 
of $\mathbb{N}$ with $+$ and $<$, for example, is decidable.

Different kinds of automata encode different kinds of information, and
in this article we will be interested in \textbf{Mealey Automata} which encode
functions from a set of words to itself. Indeed, the functions we consider
will all be invertible (and the inverses are comutable by automata as well),
and thus they will generate \textbf{Automata Groups}.
These groups are surprisingly complicated, and a classification of all groups 
generated by three state automata over the alphabet $\2 = \{0,1\}$ is an 
extremely difficult problem, though much impressive progress has been 
made \cite{Bondarenko09:three_state}. This complexity can be useful, 
as automata groups have become a rich source of examples and counterexamples
in group theory
\cite{Nekrashevych05:self_similar_groups%
     ,Sidki00:one_rooted_trees%
     ,GrigorchukNS00:automata_groups%
     }. 
Most notably, automata groups provide examples of finitely generated 
infinite torsion groups, with application to 
Burnside's Problem \cite{Gupta83:burnside}, and automata groups have
provided the only examples of groups of intermediate growth, providing 
counterexamples to Milnor's Conjecture regarding the existence of such groups
\cite{Grigorchuk11:Milnor}. In fact, one of the simplest conceivable automata 
(shown below) already generates the lamplighter group $\Z/2\Z \wr \Z$, as
is shown in \cite{GrigorchukZuk01:lamplighter}.

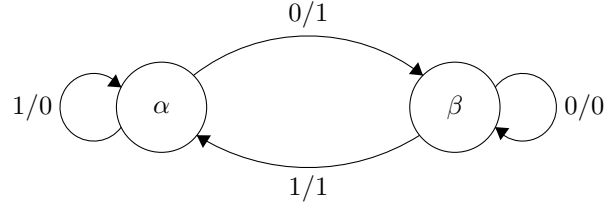
\begin{figure}
\begin{center}
\begin{tikzpicture}[scale=0.2]
\tikzstyle{every node}+=[inner sep=0pt]
\draw [black] (15.8,-29.2) circle (3);
\draw (15.8,-29.2) node {$\alpha$};
\draw [black] (35.3,-29.2) circle (3);
\draw (35.3,-29.2) node {$\beta$};
\draw [black] (13.12,-30.523) arc (-36:-324:2.25);
\draw (8.55,-29.2) node [left] {$1/0$};
\fill [black] (13.12,-27.88) -- (12.77,-27) -- (12.18,-27.81);
\draw [black] (17.918,-27.085) arc (128.02449:51.97551:12.39);
\fill [black] (33.18,-27.09) -- (32.86,-26.2) -- (32.24,-26.99);
\draw (25.55,-23.96) node [above] {$0/1$};
\draw [black] (37.98,-27.877) arc (144:-144:2.25);
\draw (42.55,-29.2) node [right] {$0/0$};
\fill [black] (37.98,-30.52) -- (38.33,-31.4) -- (38.92,-30.59);
\draw [black] (32.957,-31.065) arc (-57.68311:-122.31689:13.856);
\fill [black] (18.14,-31.06) -- (18.55,-31.91) -- (19.09,-31.07);
\draw (25.55,-33.71) node [below] {$1/1$};
\end{tikzpicture}
\end{center}

\caption{An automaton generating the lampligher group}
\label{fig:1}
\end{figure}

\subsection{Some Important Definitions}

Recall $\2 = \{0,1\}$. For our purposes, a \textbf{Mealey Automaton} is a 
tuple $\A = (S, \tau)$ where $S$ is the \textbf{State Set}, and 
$\tau : S \times \2 \to S \times \2$ is the \textbf{Transition Function}. 
We represent $\A$ as a (directed, multi-)graph with $S$ as vertices and an
edge from $s_1$ to $s_2$ (labeled by $a/b$) exactly when $\tau(s_1,a) = (s_2,b)$.
Following Sutner, whenever we have two parallel edges labeled $0/0$ and $1/1$,
we instead write one unlabeled edge to remove clutter.
We write $(\del_a s, \underline{s}(a)) = \tau(s,a)$ and call $\del_0 s$ 
(resp. $\del_1 s$) the \textbf{0-residual} (resp. \textbf{1-residual}) of $s$.

We can extend $\underline{s}$ to a length preserving function on the 
free monoid $\2^*$ as follows (here juxtaposition is concatenation, 
and the empty word $\varepsilon$ is the identity):

$\underline{s} : \2^* \to \2^*$
\begin{align*}
  \underline{s}(\varepsilon) &= \varepsilon\\
  \underline{s}(ax)       &= a' \underline{s'}(x) 
  ~~~(\text{where } (s', a') = \tau(s,a))
\end{align*}

\begin{figure}[H]
  \centering
  \begin{subfigure}{.16\textwidth}
    \centering
        \begin{tikzpicture}[scale=0.07]
        \tikzstyle{every node}+=[inner sep=0pt]
        \draw [black, fill=blue!30] (26.6,-12.1) circle (3);
        \draw (26.6,-12.1) node {$x$};
        \draw [black] (17.4,-25.6) circle (3);
        \draw (17.4,-25.6) node {$y$};
        \draw [black] (35.7,-25) circle (3);
        \draw (35.7,-25) node {$z$};
        \draw [black] (16.648,-22.708) arc (-174.31762:-254.2298:9.658);
        \fill [black] (16.65,-22.71) -- (17.07,-21.86) -- (16.07,-21.96);
        \draw (17.67,-14.97) node [left] {$1/0$};
        \draw [black] (33.761,-27.277) arc (-48.14861:-128.09563:11.117);
        \fill [black] (33.76,-27.28) -- (32.83,-27.44) -- (33.5,-28.18);
        \draw [black] (29.501,-12.822) arc (67.8011:2.5992:10.451);
        \fill [black] (29.5,-12.82) -- (30.05,-13.59) -- (30.43,-12.66);
        \draw [black] (32.758,-24.474) arc (-108.88288:-180.71682:9.832);
        \fill [black] (32.76,-24.47) -- (32.16,-23.74) -- (31.84,-24.69);
          \draw (27.31,-22.21) node [left] {$0/1$};
        \end{tikzpicture}
    \caption{$\underline{x}(0110)$}
  \end{subfigure}%
  \begin{subfigure}{.16\textwidth}
    \centering
        \begin{tikzpicture}[scale=0.07]
        \tikzstyle{every node}+=[inner sep=0pt]
        \draw [black] (26.6,-12.1) circle (3);
        \draw (26.6,-12.1) node {$x$};
        \draw [black] (17.4,-25.6) circle (3);
        \draw (17.4,-25.6) node {$y$};
        \draw [black, fill=blue!30] (35.7,-25) circle (3);
        \draw (35.7,-25) node {$z$};
        \draw [black] (16.648,-22.708) arc (-174.31762:-254.2298:9.658);
        \fill [black] (16.65,-22.71) -- (17.07,-21.86) -- (16.07,-21.96);
        \draw (17.67,-14.97) node [left] {$1/0$};
        \draw [black] (33.761,-27.277) arc (-48.14861:-128.09563:11.117);
        \fill [black] (33.76,-27.28) -- (32.83,-27.44) -- (33.5,-28.18);
        \draw [black] (29.501,-12.822) arc (67.8011:2.5992:10.451);
        \fill [black] (29.5,-12.82) -- (30.05,-13.59) -- (30.43,-12.66);
        \draw [black] (32.758,-24.474) arc (-108.88288:-180.71682:9.832);
        \fill [black] (32.76,-24.47) -- (32.16,-23.74) -- (31.84,-24.69);
        \draw (27.31,-22.21) node [left] {$0/1$};
        \end{tikzpicture}
    \caption{$1 \underline{z}(110)$}
  \end{subfigure}%
  \begin{subfigure}{.16\textwidth}
    \centering
        \begin{tikzpicture}[scale=0.07]
        \tikzstyle{every node}+=[inner sep=0pt]
        \draw [black, fill=blue!30] (26.6,-12.1) circle (3);
        \draw (26.6,-12.1) node {$x$};
        \draw [black] (17.4,-25.6) circle (3);
        \draw (17.4,-25.6) node {$y$};
        \draw [black] (35.7,-25) circle (3);
        \draw (35.7,-25) node {$z$};
        \draw [black] (16.648,-22.708) arc (-174.31762:-254.2298:9.658);
        \fill [black] (16.65,-22.71) -- (17.07,-21.86) -- (16.07,-21.96);
        \draw (17.67,-14.97) node [left] {$1/0$};
        \draw [black] (33.761,-27.277) arc (-48.14861:-128.09563:11.117);
        \fill [black] (33.76,-27.28) -- (32.83,-27.44) -- (33.5,-28.18);
        \draw [black] (29.501,-12.822) arc (67.8011:2.5992:10.451);
        \fill [black] (29.5,-12.82) -- (30.05,-13.59) -- (30.43,-12.66);
        \draw [black] (32.758,-24.474) arc (-108.88288:-180.71682:9.832);
        \fill [black] (32.76,-24.47) -- (32.16,-23.74) -- (31.84,-24.69);
        \draw (27.31,-22.21) node [left] {$0/1$};
        \end{tikzpicture}
    \caption{$11 \underline{x}(10)$}
  \end{subfigure}%
  \begin{subfigure}{.16\textwidth}
    \centering
      \begin{tikzpicture}[scale=0.07]
        \tikzstyle{every node}+=[inner sep=0pt]
        \draw [black] (26.6,-12.1) circle (3);
        \draw (26.6,-12.1) node {$x$};
        \draw [black, fill=blue!30] (17.4,-25.6) circle (3);
        \draw (17.4,-25.6) node {$y$};
        \draw [black] (35.7,-25) circle (3);
        \draw (35.7,-25) node {$z$};
        \draw [black] (16.648,-22.708) arc (-174.31762:-254.2298:9.658);
        \fill [black] (16.65,-22.71) -- (17.07,-21.86) -- (16.07,-21.96);
        \draw (17.67,-14.97) node [left] {$1/0$};
        \draw [black] (33.761,-27.277) arc (-48.14861:-128.09563:11.117);
        \fill [black] (33.76,-27.28) -- (32.83,-27.44) -- (33.5,-28.18);
        \draw [black] (29.501,-12.822) arc (67.8011:2.5992:10.451);
        \fill [black] (29.5,-12.82) -- (30.05,-13.59) -- (30.43,-12.66);
        \draw [black] (32.758,-24.474) arc (-108.88288:-180.71682:9.832);
        \fill [black] (32.76,-24.47) -- (32.16,-23.74) -- (31.84,-24.69);
        \draw (27.31,-22.21) node [left] {$0/1$};
        \end{tikzpicture}
    \caption{$110 \underline{y}(0)$}
  \end{subfigure}%
  \begin{subfigure}{.16\textwidth}
    \centering
        \begin{tikzpicture}[scale=0.07]
        \tikzstyle{every node}+=[inner sep=0pt]
        \draw [black] (26.6,-12.1) circle (3);
        \draw (26.6,-12.1) node {$x$};
        \draw [black] (17.4,-25.6) circle (3);
        \draw (17.4,-25.6) node {$y$};
        \draw [black, fill=blue!30] (35.7,-25) circle (3);
        \draw (35.7,-25) node {$z$};
        \draw [black] (16.648,-22.708) arc (-174.31762:-254.2298:9.658);
        \fill [black] (16.65,-22.71) -- (17.07,-21.86) -- (16.07,-21.96);
        \draw (17.67,-14.97) node [left] {$1/0$};
        \draw [black] (33.761,-27.277) arc (-48.14861:-128.09563:11.117);
        \fill [black] (33.76,-27.28) -- (32.83,-27.44) -- (33.5,-28.18);
        \draw [black] (29.501,-12.822) arc (67.8011:2.5992:10.451);
        \fill [black] (29.5,-12.82) -- (30.05,-13.59) -- (30.43,-12.66);
        \draw [black] (32.758,-24.474) arc (-108.88288:-180.71682:9.832);
        \fill [black] (32.76,-24.47) -- (32.16,-23.74) -- (31.84,-24.69);
        \draw (27.31,-22.21) node [left] {$0/1$};
        \end{tikzpicture}
    \caption{$1100 \underline{z}(\varepsilon)$}
  \end{subfigure}%
  \begin{subfigure}{.16\textwidth}
    \centering
        \begin{tikzpicture}[scale=0.07]
        \tikzstyle{every node}+=[inner sep=0pt]
        \draw [black] (26.6,-12.1) circle (3);
        \draw (26.6,-12.1) node {$x$};
        \draw [black] (17.4,-25.6) circle (3);
        \draw (17.4,-25.6) node {$y$};
        \draw [black] (35.7,-25) circle (3);
        \draw (35.7,-25) node {$z$};
        \draw [black] (16.648,-22.708) arc (-174.31762:-254.2298:9.658);
        \fill [black] (16.65,-22.71) -- (17.07,-21.86) -- (16.07,-21.96);
        \draw (17.67,-14.97) node [left] {$1/0$};
        \draw [black] (33.761,-27.277) arc (-48.14861:-128.09563:11.117);
        \fill [black] (33.76,-27.28) -- (32.83,-27.44) -- (33.5,-28.18);
        \draw [black] (29.501,-12.822) arc (67.8011:2.5992:10.451);
        \fill [black] (29.5,-12.82) -- (30.05,-13.59) -- (30.43,-12.66);
        \draw [black] (32.758,-24.474) arc (-108.88288:-180.71682:9.832);
        \fill [black] (32.76,-24.47) -- (32.16,-23.74) -- (31.84,-24.69);
        \draw (27.31,-22.21) node [left] {$0/1$};
        \end{tikzpicture}
    \caption{$1100$}
  \end{subfigure}

  \caption%
  {%
    An example computation -- $\underline{x}(0110) = 1100$. 
  }
\end{figure}
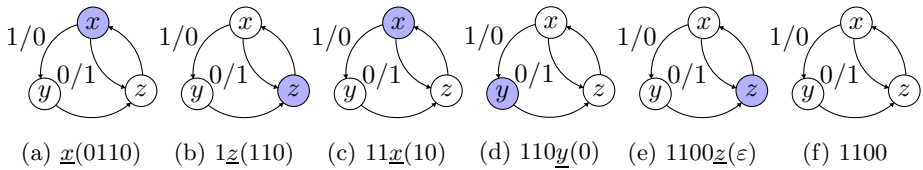

Clearly we can treat $\underline{s}$ as a function 
on $\2^\omega$, the set of infinite words, instead. In this case, automata 
provide a computable way of encoding complicated continuous functions from 
cantor space to itself, with ties to descriptive set theory%
\cite{skrzypczak15:descriptive}. If all of these functions are invertible, 
we let $\G(\A)$ denote the group generated by these functions. We write our
group additively, and denote the identity by $I$.

We can extend the definition of residuals to the whole group $\G(\A)$ by
defining the \textbf{0-residual} (resp. \textbf{1-residual}) of a 
function $f \in \G(\A)$ as the unique function 
$\del_0 f$ such that for all $w$, $f(0w) = f(0) \del_0 f(w)$ 
(resp. $f(1w) = f(1) \del_1 f(w)$). 
For a state $s \in S$, it is clear that 
$\del_a \underline{s} = \underline{s'}$, where $(s',a') = \tau(s,a)$, so
that this extends the old definition. 

Thus $\G(\A)$ can also be viewed as an automaton, by taking $\G(\A)$ 
as the state set and defining $\tau(f,i) = (\del_i f, f(i))$. Under this
definition, we find $\A$ as a natural subautomaton of $\G(\A)$ by identifying 
$s \in \A$ with $\underline{s} \in \G(\A)$.
We will call a function \textbf{Odd} if it flips its first bit, and 
\textbf{Even} otherwise, and we call an automaton \textbf{Abelian} or 
\textbf{Trivial} exactly when its group is. 

So in figure $\ref{fig:1}$, $\underline{\alpha}$ is odd, $\underline{\beta}$ 
is even, $\del_0 \underline{\alpha} = \underline{\beta}$, 
and $\del_1 \underline{\alpha} = \underline{\alpha}$.
For a more in depth description of Mealy Automata and their properties, 
see \cite{Sakarovitch09:automata_theory,Holcombe,SutnerLewi12:iter_inver_bin_trans}

The results of this paper stand on the shoulders of a result of Nekrashevich and 
Sidki that every abelian automata group is either torsion free abelian or 
boolean \cite{NekrashevychSidki04:automorphisms}. Because of this classification, 
much of the interesting structure of these groups comes from their residuation
functions. To that end, for the duration of this paper, 
homomorphisms and isomorphisms are all restricted to those 
which preserve the residuation structure in addition to the group structure.
It is a theorem by Sutner \cite{Sutner18:abelian_automata} 
that $\G(\A)$ is abelian iff for even states $\del_1 f - \del_0 f = I$ 
and for odd states $\del_1 f - \del_0 f = \gamma$, where $\gamma$ is 
independent of $f$. Moreover, the case $\gamma = I$ corresponds 
precisely to the case where $\G(\A)$ is boolean.
We now restrict ourselves further to the case where $\G(\A)$ is 
free abelian, that is to say $\G(\A) \cong \Z^m$ for some $m$,
and $\gamma \not = I$.%
\footnote%
{%
  For historical reasons we use $\Z^m$ instead of $\Z^n$ because 
  traditionally $n$ is reserved for the size of the state set of 
  an automaton.
}

\subsection{The Complete Automaton}

From the discussion above, it follows that $\Z^m \cong \G(\A)$ carries a
residuation structure, and Nekrashvych and Sidki also give a characterization 
of all possible such structures \cite{NekrashevychSidki04:automorphisms}.

Without loss of generality, we can take the odd (resp. even) states to be 
exactly the vectors with odd (resp. even) first component. The automata 
structure is given by the following affine maps
(which depend on a matrix $\Am$ and an odd vector $\e$):

\begin{align}
  \tau_{\Am, \e}(\vv, 0) &= 
  \begin{cases} 
    (\Am \vv, 0)        & \vv \text{ even}\\
    (\Am (\vv - \e), 1) & \vv \text{ odd}
  \end{cases}
  \label{eq:1}\\
  \tau_{\Am, \e}(\vv, 1) &=
  \begin{cases} 
    (\Am \vv, 1)        & \vv \text{ even}\\
    (\Am (\vv + \e), 0) & \vv \text{ odd}
  \end{cases}
  \label{eq:2}
\end{align}

In the above definition, $\Am$ is a ``$\frac{1}{2}$--integral'' matrix $\Am$ of 
$\Q$-irreducible character. This group (with its residuation structure) is 
generated by a \emph{finite} automaton exactly when $\Am$ is a contraction 
(that is, all of its complex eigenvalues have norm $<1$). By a 
$\frac{1}{2}$--integral matrix, we mean a matrix of the form

\[
\begin{pmatrix}
  \frac{a_{11}}{2} & a_{12} & \dots  & a_{1n}\\
  \vdots           & \vdots & \ddots & \vdots\\
  \frac{a_{n1}}{2} & a_{n2} & \dots  & a_{nn}\\
\end{pmatrix}
\]

\noindent
where each $a_{ij} \in \Z$. 

These matrices all have characteristic polynomial
$\chi = x^n + \frac{1}{2}g(x)$, where $g \in \Z[x]$ and has constant term 
$\pm 1$. Without loss of generality we may take $\Am$ to be in rational 
canonical form, with the coefficients of $\chi$ in the first column.

Now we see the reason for the \textbf{Residuation Vector} $\e$ in the 
above definition. Since $\Am : 2\Z \oplus \Z^{m-1} \to \Z^m$, the transition
function $\tau$ can act as simple multiplication on even vectors. However,
to ensure we have an integral output, we must first make an odd vector 
even by adding or subtracting another odd vector $\e$. It is easy to see 
that this definition gives rise to the following residuation structure:

If $\vv$ is even:
\[ \del_0 \vv = \del_1 \vv = \Am \vv \]

If $\vv$ is odd:
\[ \del_0 \vv = \Am (\vv - \e) \]
\[ \del_1 \vv = \Am (\vv + \e) \]

\noindent
Following Sutner \cite{Sutner18:abelian_automata}, for a specific matrix
$\Am$ and residuation vector $\e$ we define the \textbf{Complete Automaton} 
$\C = (\Z^m, \tau_{\Am,\e})$. If $\G(\A) \cong \C$, we say a 
function $f \in \G(\A)$ is \textbf{Located at} $\vv \in \C$ iff the isomorphism 
between $\G(\A)$ and $\C$ sends $f$ to $\vv$. Finally, given any 
state $\vv \in \C$, closing $\{ \vv \}$ under residuation will result in an
automaton $\A_{\vv}$ (which will be finite whenever $\Am$ is contracting). 
We say $\A$ is \textbf{Located at} $\vv \in \C$ iff the isomorphism sends 
$\A \subseteq \G(\A)$ to $\A_{\vv} \subseteq \C$. 

Keep in mind the distinction between the group of functions $\G(\A)$
and a particular isomorphism between $\G(\A)$ and $\C$. We will freely
identify these objects, but the location of a particular function
depends heavily on the choice of $\e$.

Nekrashevich and Sidki's theorem gives us a purely linear algebraic method
for discussing these automata groups, since a restatement of their theorem 
says that every torsion free abelian automata group $\G(\A)$ is isomorphic 
to\footnote{Recall our isomorphisms preserve the resituation structure in addition to the group structure} 
$\C$ for some $\Am$ and $\e$. Seeing this fact, it is natural to ask if,
given an automaton $\A$, we can characterize all $\Am$ and $\e$ for which 
$\A \subseteq \C$. Indeed, it is natural to ask which vector $\vv$ will $\A$
be located at in this identification.

Nekrashevych and Sidki show that each $\A$ 
has a unique matrix $\Am$ (up to GL($\Q$) similarity) which works, though
their proof is nonconstructive. We call this $\Am$ (in rational canonical form)
the \textbf{Associated Matrix} of $\A$. Unfortuately, Nekrashevych and Sidki 
leave entirely open the question of which $\e$ admit $\A$ as a subautomaton
once we have the correct $\Am$, and moreover where $\A$ is located
if an embedding into $\C$ exists. Algorithms for determining the matrix from
the automaton are given by Okano \cite{Okano15:thesis} 
and Becker \cite{Becker18:thesis}, solving part of the problem. 

In this paper we finish the job by fully
characterizing the impact of $\e$ on the residuation structure of $\C$.
For a more detailed discussion of these linear algebraic methods and their 
origins, see %
\cite{Nekrashevych05:self_similar_groups,NekrashevychSidki04:automorphisms}.

\subsection{Principal Automata}
Each abelian automaton gets a unique associated matrix as above, but each
matrix can be associated to infinitely many automata.
It was shown by Okano \cite{Okano15:thesis} that there is a 
distinguished automaton, now called the \textbf{Principal Automaton} $\P$, 
associated to each matrix. $\P(\Am)$ is defined to be 
$\P = \A_{\e_1} \cup \A_{-\e_1} \subseteq \mathfrak{C}(\Am, \e_1)$,
though there is a longstanding conjecture (introduced in the same paper) 
that in most cases this is the same machine as 
$\A_{\e_1} \subseteq \mathfrak{C}(\Am, \e_1)$. We will write $\P$ whenever $\Am$
is clear from context. The function located at $\e_1 \in \mathfrak{C}(\Am, \e_1)$ 
will be important later on, and so we write $\delta$ for this function. Notice this 
means $\P$ is the smallest automaton containing $\delta$ and $-\delta$.
We will write $\P$ when the matrix is clear from context. 

As we will see, $\delta$ is located at $\e \in \C$ for all $\e$. This will 
give us a way to compare functions in various $\C$ by using $\delta$ as a 
kind of meterstick. Indeed, since $\P$ is generated by $\pm \delta$, the
next theorem will show that for every automaton $\A$ with associated matrix $\Am$, 
$\G(\P(\Am)) \leq \G(\A)$. Thus every function in $\P$ is a $\Z$-linear 
combination of functions in $\A$. In particular, we see
$\P$ is a subautomaton of $\G(\A)$ for every $\A$ with matrix $\Am$.
While there are proofs of this claim which rely heavily on the
ambient linear algebraic structure \cite{Okano15:thesis}, 
we present here a construction which uses only the given 
automaton $\A$ to construct $\P$. Thus every $\underline{s} \in \G(\P)$ 
is already in $\G(\A)$, and the subgroup relation follows.

\begin{thm}
  For each nontrivial $\A$ with associated matrix $\Am$, $\G(\P) \leq \G(\A)$.
\end{thm}

\begin{proof}
  It was shown in \cite{Sutner18:abelian_automata} that $\gamma$ depends
  only on the matrix $\Am$, so that for any automata $\A$ and $\A'$ with the same
  associated matrix $\Am$, and for any odd states $f \in \A$, $f' \in \A'$, 
  we have $\gamma = \del_1 f - \del_0 f = \del_1 f' - \del_0 f'$. 
  In particular, for $\delta \in \P$, we have 
  $\gamma = \del_1 \delta - \del_0 \delta = \del_1 \delta$ since
  $\del_0 \delta = \Am(\e_1 - \e_1) = \bar{0} = I$. 

  Since we know from the previous discussion that $\P$ is generated by
  $\pm \delta$, we can build it by hand by leveraging the fact that 
  $\gamma = \del_1 \delta$ is already in $\G(\A)$.

  Let $\A$ be an abelian automaton with at least one odd state.
  Note that if $\A$ has no odd states, its group is trivial, so we may
  safely ignore it.
  Put $\gamma = \del_1 o - \del_0 o$ for $o \in \A$ odd, and construct
  a new automaton by closing $\gamma$ under residuation.
  Note that this can be done using only information contained in $\A$,
  since it is easy to check that:
  \[
    \del_0(f + g) = \begin{cases} \del_0 f + \del_1 g & \text{both odd}\\
                                  \del_0 f + \del_0 g & \text{otherwise}
                    \end{cases}
  \]
  \[
    \del_1(f + g) = \begin{cases} \del_1 f + \del_0 g & \text{both odd}\\
                                  \del_1 f + \del_1 g & \text{otherwise}
                    \end{cases}
  \]
  \[
    \del_0 (-f) = - \del_1 f
  \]
  \[
    \del_1 (-f) = - \del_0 f
  \]

  Thus using the characterization by Sutner \cite{Sutner18:abelian_automata}
  that a state is odd iff it has distinct residuals, we can close $\gamma$ 
  under residuation using only information in $\A$.
  Since $\gamma \in \G(\A)$ and $\G(\A)$ is residuation closed, 
  this entire closure is a subset of $\G(\A)$. Morever, whenever $\A$ is 
  finite, $\Am$ is contracting and so $\P$ is finite too. Thus this 
  procedure can actually be carried out.

  Another theorem by Sutner \cite{Sutner18:abelian_automata} says that 
  $\G(\A_{\vv}) = \G(\A_{\overline{w}})$ whenever $\overline{w}$ transitions 
  into $\vv$. Because of this, the above closure generates the same group
  as the above closure with an additional state ($\delta$) residuating into 
  $\gamma$ and a self loop ($I$). This new machine is exactly
  $\A_{\e_1} \subseteq \mathfrak{C}(\Am,\e_1)$. Any state in $\A_{\-e_1}$ is
  the negation of a state in $\A_{e_1}$, and so 
  $\P(\Am) = \A_{\e_1} \cup \A_{-\e_1} \subseteq \G(\A)$. 
  Then $\G(\P) \leq \G(\A)$, as desired.
\end{proof}

\subsection{An Example}
Consider the following machine, $\A^3_2$:

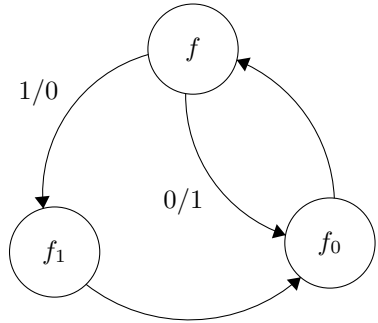
\begin{center}
\begin{tikzpicture}[scale=0.2]
\tikzstyle{every node}+=[inner sep=0pt]
\draw [black] (26.6,-12.1) circle (3);
\draw (26.6,-12.1) node {$f$};
\draw [black] (17.4,-25.6) circle (3);
\draw (17.4,-25.6) node {$f_1$};
\draw [black] (35.7,-25) circle (3);
\draw (35.7,-25) node {$f_0$};
\draw [black] (16.648,-22.708) arc (-174.31762:-254.2298:9.658);
\fill [black] (16.65,-22.71) -- (17.07,-21.86) -- (16.07,-21.96);
\draw (17.67,-14.97) node [left] {$1/0$};
\draw [black] (33.761,-27.277) arc (-48.14861:-128.09563:11.117);
\fill [black] (33.76,-27.28) -- (32.83,-27.44) -- (33.5,-28.18);
\draw [black] (29.501,-12.822) arc (67.8011:2.5992:10.451);
\fill [black] (29.5,-12.82) -- (30.05,-13.59) -- (30.43,-12.66);
\draw [black] (32.758,-24.474) arc (-108.88288:-180.71682:9.832);
\fill [black] (32.76,-24.47) -- (32.16,-23.74) -- (31.84,-24.69);
\draw (27.31,-22.21) node [left] {$0/1$};
\end{tikzpicture}
\end{center}

As before, the unlabeled transitions both copy the input bit, however these
have been omitted for cleanliness.

Then by letting $\gamma = \del_1 f - \del_0 f = f_1 - f_0$, and closing
under residuation using the above algorithm, we construct the following 
machine ($\gamma$ is shown at the bottom left): 

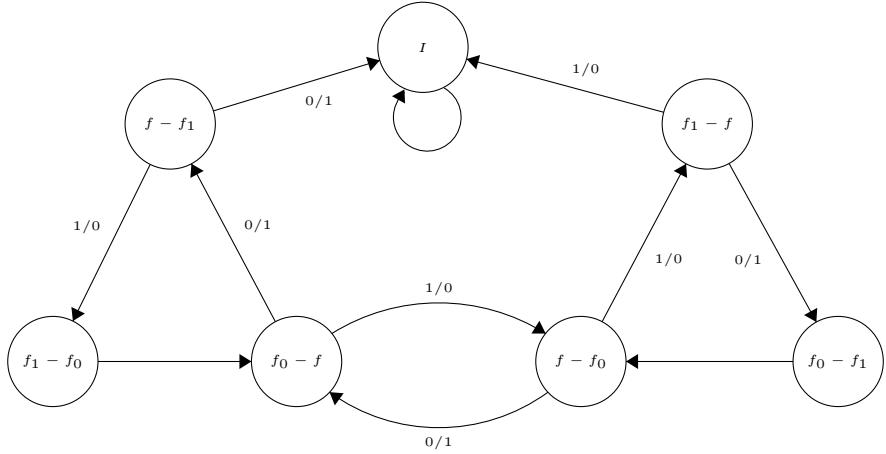
\begin{center}
\begin{tikzpicture}[scale=0.2]
\tikzstyle{every node}+=[inner sep=0pt]
\draw [black] (36.4,-15.9) circle (3);
\draw (36.4,-15.9) node {\tiny $I$};
\draw [black] (37.8,-18.6) arc (60.34019:-227.65981:2.25);
\fill [black] (35.17,-18.7) -- (34.34,-19.15) -- (35.21,-19.64);
\draw [black] (55.3,-21) circle (3);
\draw (55.3,-21) node {\tiny $f_1-f$};
\draw [black] (64,-36.8) circle (3);
\draw (64,-36.8) node {\tiny $f_0-f_1$};
\draw [black] (46.9,-36.8) circle (3);
\draw (46.9,-36.8) node {\tiny $f-f_0$};
\draw [black] (28,-36.8) circle (3);
\draw (28,-36.8) node {\tiny $f_0-f$};
\draw [black] (11.8,-36.8) circle (3);
\draw (11.8,-36.8) node {\tiny $f_1-f_0$};
\draw [black] (19.6,-21) circle (3);
\draw (19.6,-21) node {\tiny $f-f_1$};
\draw [black] (22.47,-20.13) -- (33.53,-16.77);
\fill [black] (33.53,-16.77) -- (32.62,-16.53) -- (32.91,-17.48);
\draw (29.5,-19.03) node [below] {\tiny $0/1$};
\draw [black] (18.27,-23.69) -- (13.13,-34.11);
\fill [black] (13.13,-34.11) -- (13.93,-33.61) -- (13.03,-33.17);
\draw (15,-27.81) node [left] {\tiny $1/0$};
\draw [black] (14.8,-36.8) -- (25,-36.8);
\fill [black] (25,-36.8) -- (24.2,-36.3) -- (24.2,-37.3);
\draw [black] (26.59,-34.15) -- (21.01,-23.65);
\fill [black] (21.01,-23.65) -- (20.94,-24.59) -- (21.83,-24.12);
\draw (24.48,-27.74) node [right] {\tiny $0/1$};
\draw [black] (30.345,-34.939) arc (122.03013:57.96987:13.397);
\fill [black] (44.56,-34.94) -- (44.14,-34.09) -- (43.61,-34.94);
\draw (37.45,-32.4) node [above] {\tiny $1/0$};
\draw [black] (52.4,-20.22) -- (39.3,-16.68);
\fill [black] (39.3,-16.68) -- (39.94,-17.37) -- (40.2,-16.41);
\draw (47.22,-17.84) node [above] {\tiny $1/0$};
\draw [black] (56.75,-23.63) -- (62.55,-34.17);
\fill [black] (62.55,-34.17) -- (62.61,-33.23) -- (61.73,-33.71);
\draw (58.98,-30.09) node [left] {\tiny $0/1$};
\draw [black] (61,-36.8) -- (49.9,-36.8);
\fill [black] (49.9,-36.8) -- (50.7,-37.3) -- (50.7,-36.3);
\draw [black] (48.31,-34.15) -- (53.89,-23.65);
\fill [black] (53.89,-23.65) -- (53.07,-24.12) -- (53.96,-24.59);
\draw (51.78,-30.06) node [right] {\tiny $1/0$};
\draw [black] (44.711,-38.841) arc (-53.966:-126.034:12.343);
\fill [black] (30.19,-38.84) -- (30.54,-39.72) -- (31.13,-38.91);
\draw (37.45,-41.7) node [below] {\tiny $0/1$};
\end{tikzpicture}
\end{center}

It is easy to check that this is the principal machine for
$\Am = \begin{pmatrix} -1 & 1 \\ -\frac{1}{2} & 0 \end{pmatrix}$,
where $f - f_1 = \delta$ is located at $\e_1 \in \mathfrak{C}(\Am, \e_1)$.
Moreover, one can check that $f \in \A^3_2$ as above is located at 
$\e_1 \in \mathfrak{C}\left ( \Am, \begin{pmatrix} 3 \\ 2 \end{pmatrix} \right )$.
Thus, $f - f_1$ is located at 
\[ 
  \e_1 - \del_1 \e_1 = 
  \e_1 - \Am \left ( \e_1 + \begin{pmatrix} 3 \\ 2 \end{pmatrix} \right ) = 
  \e_1 - \begin{pmatrix} -2 \\ -2 \end{pmatrix} =
  \begin{pmatrix} 3 \\ 2 \end{pmatrix} \in
  \mathfrak{C}\left ( \Am, \begin{pmatrix} 3 \\ 2 \end{pmatrix} \right ).
\]

When running the algorithm in this case, we do not need to separately add
$\pm \delta$ or the inverse machine. Here $\delta$ is already in the closure of
$\gamma$ under residuation, and the machine is already closed under negation. 
The Strongly Connected Component Conjecture 
predicts that this will be the case whenever $\Am$ has characteristic 
polynomial other than $x^m - \frac{1}{2}$, which corresponds to the so called 
sausage automata. Unfortunately, this conjecture is yet unproven,
and so in the above proof we had to explicitly add in these extra states.

\section{Fractional Extensions}
Going forward, $\G = \mathfrak{C}(\Am,\e_1)$ will denote $\G(\P)$ for some 
principal machine $\P$. 

Since $\Am$ sends $2\Z \oplus \Z^{m-1}$ to $\Z^m$, 
$\Am^{-1}$ sends $\Z^m$ to $2\Z \oplus \Z^{m-1}$, and so has 
only integer entries. Thus we can give $\G$ the structure of a 
$\Z[x]$ module where $x \cdot \vv = \Am^{-1}\vv$, extended linearly. Further, 
since $\Am$ has irreducible characteristic polynomial so does $\Am^{-1}$. Thus 
this module is cyclic, and is generated by $\e_1 = \delta$. The cyclicity of 
this module tells us that we can identify our states $\Z^m$ with $\Z[x] / \chi^*$ 
where $\chi^*$ is the characteristic polynomial of $\Am^{-1}$ and has degree 
$m$. This identifies a vector $\vv$ with the polynomial $p_{\vv}$ whose 
coefficients are the coordinates of $\vv$ (the constant term is the first 
coefficient). Said another way, $\vv = p_{\vv} \cdot \e_1$.

Now for $p \in \Z[x]$ with odd constant term, we write
$p^{-1} \cdot \G$ in place of $\G(\mathfrak{C}(\Am, p \cdot \e_1))$.
That is to say, $p^{-1} \cdot \G$ has as its states $\Z^m$ and as its 
odd residuations
$\del_0 \vv = \Am (\vv - p \cdot \e_1)$, and 
$\del_1 \vv = \Am (\vv + p \cdot \e_1)$.
We will only discuss polynomials $p$ with an odd constant term, as 
this ensures $p \cdot \e_1$, our residuation vector, is odd.

We call $p^{-1} \cdot \G$ the \textbf{Fractional Extension} of $\G$ by $p$.
To first justify the use of the word ``extension'', notice 
$\G \hookrightarrow p^{-1} \cdot \G$ for all $p$ by the
homomorphism $\vv \mapsto p \cdot \vv$. 
Further, if $p$ is not a unit in $\Z[x] / \chi^*$, this 
homomorphism is \emph{not} surjective. That is to say $\G$ is a proper 
subgroup of $p^{-1} \cdot \G$.
This observation is true in more generality, as shown below.
Recall we work in the category whose arrows also preserve the residuation
structure, and thus an isomorphism identifying two vectors will show that
those vectors compute the same function on $\2^\omega$. However, it means
we must show that our embeddings genuinely do preserve this structure.

\begin{thm}
  If $rp = q$ in $\Z[x] / \chi^*$, then 
  $p^{-1} \cdot \G \hookrightarrow q^{-1} \cdot \G$, 
  with a canonical injection $\varphi_r : \vv \mapsto r \cdot \vv$. 
  In particular, if $r$ is a unit, then $p^{-1} \cdot \G \cong q^{-1} \cdot \G$.
\end{thm}

\begin{proof}
  It is clear that $\varphi_r$ preserves the $\Z[x]$-module structure, so
  it remains to show that it preserves the residuation structure.

  Let $rp = q$, $f \in p^{-1} \cdot \G$ located at $\vv$. 
  Consider $f' \in q^{-1} \cdot \G$ located at $r \cdot \vv$.
  First note $f$ and $f'$ have the same parity, since 
  $r$ has odd constant term, and so $\vv$ and $r \cdot \vv$
  have the same parity. Now, consider the residuals of $f$ and $f'$. 
  
  If $f$ is even, then 
  \[ \del_0 f' = \Am (r \cdot \vv) = r \cdot \Am \vv = r \cdot \del_0 f \]

  If $f$ is odd, then
  \[ \del_0 f' = \Am (r \cdot \vv - q \cdot \e_1) 
               = r \cdot \Am (\vv - p \cdot \e_1)
               = r \cdot \del_0 f \]
  A similar argument shows $\del_1 f' = r \cdot \del_1 f$

  If $r$ is a unit, then $r^{-1}$ also has odd constant term 
  (since $r r^{-1} = 1$ has odd constant term) and so $\varphi_r$
  is an isomorphism with inverse $\varphi_{r^{-1}}$.
\end{proof}

The previous proof has justified the use of the word ``extension'', 
but it is still not clear why this extension should be ``fractional''.
As the previous proof shows, $p \cdot \vv \in p^{-1} \cdot \G$, 
computes exactly the same function as $\vv \in \G$.
However, most vectors cannot be written as a multiple of $p$.
What do they do as functions?
We call such vectors (and their corresponding functions)
\textbf{Fractional}, due to the following analogy:

Say we are only allowed to compute with $\Z$, but we want the ability
to work with fractions. We can approximate $\Q$ by allowing fractions
with fixed denominator. If we write 
$\frac{1}{n} \mathbb{Z} = \{ \frac{k}{n} ~|~ k \in \Z \}$, 
then we clearly see $\Z \cong \frac{1}{n} \Z \leq \Q$ for every $n$. 
So if we want to be able to talk about fractions like 
$\frac{1}{3}$, we might work in the ``extension'' $\frac{1}{3} \Z$, where 
$\Z \hookrightarrow \frac{1}{3}\Z$ by the embedding $k \mapsto 3k$.
The new elements, then, are ``fractional'' in the obvious sense. 
Once we have made this identification, we can (computationally) forget
about the fact that we're working in an extension at all. The equation
$4 + 6 = 10$ remains true, we simply reinterpret this as being 
$\frac{4}{3} + \frac{6}{3} = \frac{10}{3}$. 

Scaling $\G \cong \Z^m$ by some polynomial $p^{-1}$ is \emph{exactly} the
same operation. However, in this new higher dimensional setting, we have 
more degrees of freedom, and must therefore index by polynomials $p$ instead
of mere constants. Note this explains all of the ambiguity regarding the
location of a function $f$ in an extension $\C$. Since $\e$ correpsonds to
our choice of $p^{-1}$, the change in location of $f$ is entirely analogous
to the change in position of $\frac{1}{2} \in \Q$ in various $\frac{1}{n}\Z$.
In $n=10$, $\frac{1}{2}$ shows up at $5$. In $n=6$, $\frac{1}{2}$ shows up 
at $3$. But mysteriously, $\frac{1}{2}$ doesn't appear when $n=3$\ldots 
The seemingly sporadic $\e$ in which a function $f$ can be found is explained
by exactly the same phenomenon!

Of course, there is a minimal $n$ for which $\frac{1}{2} \in \frac{1}{n}\Z$.
Moreover, once we know it, we can characterize exactly where $\frac{1}{2}$
will be in all extensions $\frac{1}{m}\Z$ where $n \mid m$. In fact, we can
do the same thing for automata.

\section{Characterizing Automata}
Since each automaton $\A$ is a subautomaton of some $\C$,
equivalently some $p^{-1} \cdot \G$, there should be a minimal $\e$ 
(up to multiplication by units) which still has $\A$ as a subautomaton. 

Notice that if we locate $\A$ at $\e_1 \in p^{-1} \cdot \G$, 
then there can be no smaller polynomial $q$ (in the division ordering)
which also places $\A$ at an integral position. The following theorem 
shows this is always possible.

\begin{thm}
  Every nontrivial abelian automaton $\A$ can be 
  located at $\e_1$ in $p^{-1} \cdot \G$ for some $p$.
\end{thm}

\begin{proof}
  It is a theorem by Sutner \cite{Sutner18:abelian_automata} that every 
  finite state abelian automaton residuates into a strongly connected component, 
  and further that this component generates the same group as the entire 
  machine. So we may, with no loss of generality, assume our machine is 
  strongly connected (that is, every state except possibly $I$ has a path to
  every other state).

  Let $f$ be an odd state in $\A$. Then at least one of $\del_0 f$ and 
  $\del_1 f$ is not equal to $f$. So there is some nontrivial cycle
  from $f$ to itself, which we can represent by a matrix equation 
  relating $\vv_f$, and $\e$. (Here $\vv_f$ is where $f$ will be located, 
  and $\e$ will be the residuation vector). 
  We can then rearrange this equation to obtain 
  $p_1(\Am)\vv_f = p_2(\Am)\e$.

  Now $p_1, p_2 \in \Z[x]$, and $\Am$ has irreducible character over $\Z$.
  Then the eigenvalues of $p(\Am)$ are precisely $p(\lambda)$
  where $\lambda$ is an eigenvalue of $\Am$, so $\Am$'s invertibility implies
  the invertibility of both $p_1(\Am)$ and $p_2(\Am)$. Thus

  \[ \e = p_2(\Am)^{-1}p_1(\Am)\vv_f \]

  Choosing $\vv_f = \e_1$ gives a value for the residuation vector $\e$,
  and (since $\G$ is cyclic as a $\Z[x]$ module) a value $\e$ induces a 
  polynomial $p_{\e}$ such that $p_{\e} \cdot \e_1 = \e$. 
  Then, by construction, $\A$ is a subautomaton of $p_e^{-1} \cdot \G$, and is 
  anchored with $f$ at $\e_1$. As desired.
\end{proof}

For any automaton $\A$, we can now completely characterize in
which $\C$ it can be located, and at what vectors.
First locate $\A$ at $\e_1 \in p^{-1} \cdot \G$, and then to locate it at
any odd vector $\vv$, scale both sides by $p_{\vv}$ to see $\A$ located at
$\vv \in p_{\vv} p^{-1} \cdot \G$. 
In the above proof, the choice of $\vv_f = \e_1$ was arbitrary, and we can
directly locate $\A$ at a different odd vector $\vv'$ by setting 
$\vv_f = \vv'$. This will give the same result as locating it at $\e_1$ and 
then multiplying by $p_{\vv'}$, again, by cyclicity.
The same observation shows that, given some polynomial $q$ 
(equivalently some vector $q \cdot \e_1$) $\A$ is located somewhere in 
$q^{-1} \cdot \G = \mathfrak{C}(\Am,q \cdot \e_1)$ if and only if $p \mid q$. 
Further, it will be located at exactly $p^{-1}q \cdot e_1$.

\subsection{An Example}
Recall the abelian automaton $\A^3_2$ from earlier in the paper:

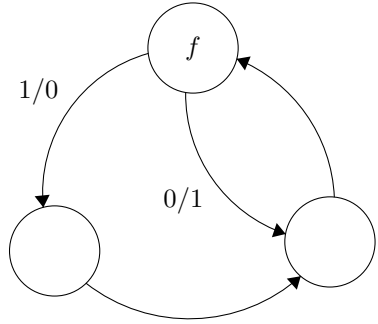
\begin{center}
\begin{tikzpicture}[scale=0.2]
\tikzstyle{every node}+=[inner sep=0pt]
\draw [black] (26.6,-12.1) circle (3);
\draw (26.6,-12.1) node {$f$};
\draw [black] (17.4,-25.6) circle (3);
\draw [black] (35.7,-25) circle (3);
\draw [black] (16.648,-22.708) arc (-174.31762:-254.2298:9.658);
\fill [black] (16.65,-22.71) -- (17.07,-21.86) -- (16.07,-21.96);
\draw (17.67,-14.97) node [left] {$1/0$};
\draw [black] (33.761,-27.277) arc (-48.14861:-128.09563:11.117);
\fill [black] (33.76,-27.28) -- (32.83,-27.44) -- (33.5,-28.18);
\draw [black] (29.501,-12.822) arc (67.8011:2.5992:10.451);
\fill [black] (29.5,-12.82) -- (30.05,-13.59) -- (30.43,-12.66);
\draw [black] (32.758,-24.474) arc (-108.88288:-180.71682:9.832);
\fill [black] (32.76,-24.47) -- (32.16,-23.74) -- (31.84,-24.69);
\draw (27.31,-22.21) node [left] {$0/1$};
\end{tikzpicture}
\end{center}

Say we want to find $\vv$ and $\e$ such that $\A^3_2$ is located at 
$\vv \in \C$.

Using the algorithm described by Becker \cite{Becker18:thesis}, we find
$\Am = \begin{pmatrix} -1 & 1 \\ -\frac{1}{2} & 0 \end{pmatrix}$.

Then notice $\del_0 \del_0 f = f$.
So $\Am^2 (\vv_f - \e) = \vv_f$, and
$\Am^2 \vv_f - \vv_f = \Am^2 \e$. Thus

\[ \e = \Am^{-2} (\Am^2 - I) \vv_f \]

Choosing $\vv_f = \e_1$ gives $\e = \begin{pmatrix} 3 \\ 2 \end{pmatrix}$.

Then $f = \begin{pmatrix} 1 \\ 0 \end{pmatrix} \in (3+2x)^{-1} \cdot \G$

\subsection{Limiting Object}
Just as we can recover $\Q$ as a limit of the fractional groups 
$\frac{1}{n} \Z$, we can define a group $\widetilde{\G}$ as the limit of
our $p^{-1} \cdot \G$. Indeed, just as
$\Q$ eliminates the parameter $n$ in $\frac{1}{n} \Z$, $\widetilde{\G}$ 
contains every automaton $\A$ at exactly one position, while removing the need
for the parameter $p$ (and thus, the parameter $\e$ in $\C$). Morever, it is
still effective to work with $\widetilde{\G}$, so we do not lose any of the
computability benefits of working with automata groups.

Formally, we define $\widetilde{\G}$ to be $\varinjlim p^{-1} \cdot \G$ 
where the colimit (in the category of abelian groups) is taken over the poset 
of polynomials (with odd constant term) under the divisbility ordering. This 
colimit can then be given residuation structure in a unique way which is 
compatable with the residuation structure on all the $p^{-1} \cdot \G$.

Explicitly, we look at the set 

\[ 
  \widetilde{\G} = \left \{ \left . \frac{\vv}{p} \ \right | \  \vv \in \Z^m, p \in \Z[x], p_0 \text{ odd} \right \} / \sim
\]

where we quotient by $\frac{\vv}{p} \sim \frac{q \cdot \vv}{pq}$ for every
$q$ with odd coefficient. We first endow this with a group structure by
declaring $[\frac{\vv}{p}] + [\frac{\overline{w}}{q}] = 
[\frac{q \cdot \vv + p \cdot \overline{w}}{pq}]$. Finally, we give
$\widetilde{\G}$ residuation structure too, by declaring 
$\del_0 [\frac{\vv}{p}] = [\frac{\del_0 \vv}{p}]$ where the residuation
on the right hand side takes place in $p^{-1} \cdot \G$. We define 
$\del_1 [\frac{\vv}{p}]$ analogously, and these operations are quickly seen
to be well defined.  This structure remains computable, because any finite 
computation we wish to do will have a least common denominator, and we can
simply work in a good enough approximation.

\section{Conclusion}
We have shown that the residuation vector $\e$ corresponds to how fine an
approximation of $\widetilde{\G}$ one wants. This is because each $\C$ 
corresponds to $p_{\e}^{-1} \cdot \G$, with progressively larger $\e$ 
corresponding to progressively more complicated fractional elements, which
approximate $\widetilde{\G}$. This allows us to characterize which automata
show up in which $\C$ (and, moreover, where they show up) by finding a minimal 
(in the division ordering) $p_{\e}$ in which a given automaton is found.

Further, the existence of the universal object $\widetilde{\G}$ 
sheds new light on the connection between affine tiles
\cite{LagariasWang96:tiles,LagariasWang97:integral_tiles}
and abelian automata noted by Sutner
\cite{Sutner18:abelian_automata}. 
Indeed it is easy to see that in 
$\widetilde{\G}$ every strongly connected component 
(and thus every subautomaton of interest) has each vector in the attractor 
of the iterated function system given by the residuation functions 
$\{ \vv \mapsto \Am \vv, \vv \mapsto \Am (\vv \pm \e_1) \}$.
Thus, in particular, the size of the principal machine is bounded by the
number of integral points in this attractor. Even in $\Z^2$, however, there
are examples where this bound is not tight.

The relation between automata and polynomials discussed in this paper 
also provides a new take on a proof technique for the longstanding
Strongly Connected Component Conjecture. This conjecture 
asserts that principal machines $\P$ have only one strongly connected component 
(plus the self looping identity state) whenever their matrix has a 
characteristic polynomial that is \emph{not} of the form $x^n + \frac{1}{2}$.
The new way of looking at residuation vectors allows us to rewrite the 
residual functions as $\del_i \vv = \Am (\vv - (-1)^i \delta)$ for $\vv$ odd.
It is easy to see, then, that the following polynomials correspond to paths
ending in $\delta$, since they undo residuation:

\begin{align*}
  P_\epsilon(x)   &= 1\\
  P_{w0}(x)       &= xP_w(x) + 0\\
  P_{w1}(x)       &= xP_w(x) + 1\\
  P_{w\bar{1}}(x) &= xP_w(x) - 1
\end{align*}

Sutner made a similar observation, and described Path Polynomials 
\cite{Sutner18:abelian_automata} which
allow us to reason about the existence of directed paths between states 
in an automaton by purely algebraic means. However, the traditional path 
polynomials are clunky and not always defined, since they correspond to paths
\emph{starting} at $\delta$, and so $P'_{w0} \cdot \delta$ is only well 
defined if $P'_w \cdot \delta$ is even (and $P'_{w1}$ and $P'_{w\bar{1}}$ 
are only well defined if $P_w \cdot \delta$ is odd). Since the polynomials 
defined above move \emph{backwards} along transitions instead of forwards, 
they are always well defined.

The existence of a path polynomial $p$ which is congruent to $-1$ mod $\chi^*$
then shows the existence of a path from $-\delta$ to $\delta$.
Then to prove the SCC conjecture, it suffices to prove that whenever $\Am$ 
does not have characteristic $x^n + \frac{1}{2}$ there is a polynomial 
$p \in \{-1,0,1\}[x]$ which is congruent to $-1$ mod $\chi^*$. 

\section*{Acknowledgements}
This paper would not exist without the advice of my advisor Klaus Sutner.
There aren't enough thanks for the hours of conversation I enjoyed.

\newpage

\bibliographystyle{splncs04}
\bibliography{bib.bib}

\end{document}